\documentclass[letterpaper, 10 pt, conference]{ieeeconf}
\IEEEoverridecommandlockouts
\let\labelindent\relax

\usepackage{xcolor}
\usepackage{amsmath,amssymb,amsthm}
\usepackage{enumitem}
\usepackage{url}
\usepackage{cleveref,cite}
\usepackage{algorithmicx,algorithm,algpseudocode}
\usepackage{graphicx,epstopdf}
\usepackage{caption,subcaption}

\newcommand{\bs}[1]{\boldsymbol{#1}}
\newcommand{\cl}[1]{\mathcal{#1}}
\newcommand{\bb}[1]{\mathbb{#1}}

\newcommand{\lb}{\left(}
\newcommand{\rb}{\right)}
\newcommand{\ls}{\left[}
\newcommand{\rs}{\right]}
\newcommand{\lc}{\left\{}
\newcommand{\rc}{\right\}}
\newcommand{\ld}{\left.}
\newcommand{\rd}{\right.}
\newcommand{\lv}{\left\vert}
\newcommand{\rv}{\right\vert}

\newcommand{\rank}{\operatorname{rank}}

\newtheorem{theorem}{Theorem}

\newtheorem{prop}{Proposition}
\newtheorem{defn}{Definition}

\newcommand{\ER}{Erd\H{o}s R\'enyi}
\title{\LARGE \bf Minimal Input Structural Modifications for Strongly Structural Controllability}

\author{Geethu Joseph$^{1}$, Shana Moothedath$^{2}$, and Jiabin Lin$^2$
\thanks{$^{1}$Geethu Joseph is with the Faculty of Electrical Engineering, Mathematics and Computer Science,
        Delft University of Technology, 2628 CD Delft, The Netherlands
        {\tt\small G.Joseph@tudelft.nl}}%
\thanks{$^{2}$Shana Moothedath and Jiabin Lin are with the Department of Electrical and Computer Engineering, Iowa State University,
        IA 50011, USA
        {\tt\small \{mshana,jiabin\}@iastate.edu}}%
}

\begin{document}
\maketitle
\thispagestyle{empty}
\pagestyle{empty}

\begin{abstract}
This paper studies the problem of modifying the input matrix of a structured system to make the system strongly structurally controllable.  We focus on the generalized structured systems that rely on zero/nonzero/arbitrary structure, i.e., some entries of system matrices are zeros, some are nonzero, and the remaining entries can be zero or nonzero (arbitrary). {We analyze the feasibility of the problem and if it is feasible, we reformulate it into another equivalent problem. This new formulation leads to a greedy heuristic algorithm}. However, we also show that the greedy algorithm can give arbitrarily poor solutions for some special systems. Our alternative approach is a randomized Markov chain Monte Carlo-based algorithm. Unlike the greedy algorithm, this algorithm is guaranteed to converge to an optimal solution with high probability. Finally, we numerically evaluate the algorithms on random graphs to show that the algorithms perform well.
\end{abstract}
\begin{keywords}
   Network controllability, pattern matrices, structured system, zero forcing, Markov chain Monte Carlo
\end{keywords}

\section{Introduction}
Network controllability is a fundamental property used to analyze the behavior of dynamical systems in complex networks like social networks~\cite{joseph2021controllability}, power systems~\cite{pequito2013framework}, and biological systems~\cite{gu2015controllability}. However, the underlying dynamical model of such systems may not be completely known. For example, in a large social network where people interact and influence each other, this influence may not be measurable~\cite{joseph2021controllability}. We may only know whether the connections exist or not, and sometimes, the interconnection structure can be unknown at several locations. Such uncertainties {in large systems like social, power, and brain networks~\cite{gu2015controllability,pequito2013framework}} can be modeled using \emph{structured systems}~\cite{mayeda1979strong}. It comprises the linear dynamical systems whose system matrices follow a \emph{zero/nonzero/arbitrary structure} (arbitrary refers to values that can be zero or nonzero). A structured system is strongly structurally controllable if all the linear dynamical systems belonging to the system are controllable. Given a strongly structurally uncontrollable system, we aim to enforce strong structural controllability by 
minimally modifying the interconnection structure between the inputs and state variables. 

Several studies have investigated the strong structural controllability of structured systems with various zero/nonzero patterns~\cite{liu2017structural,mousavi2017structural,
menara2018structural}. The related prior work addresses the problems such as graph-theoretic controllability tests~\cite{mayeda1979strong,reinschke1992strong,jarczyk2011strong}, minimum input selection~\cite{chapman2013strong,yashashwi2019minimizing,abbas2023zero,trefois2015zero}, and targeted controllability~\cite{li2020structural,monshizadeh2015strong}. However, the above works are based on structured systems defined by zero/nonzero patterns. 
Recently, this model has been generalized to accommodate a third possibility that a system matrix entry is not a fixed zero or nonzero but can take any real value~\cite{jia2020unifying}. This work presented a controllability test for the generalized systems. 

Furthermore, to the best of our knowledge, the structural modification problem has not been studied in the context of strong structural controllability for either zero/nonzero or zero/nonzero/arbitrary systems. 
{The current studies for zero/nonzero systems focus on designing a minimal input matrix for a specified state matrix, not on structural modification of a given input matrix~\cite{chapman2013strong,yashashwi2019minimizing,abbas2023zero,trefois2015zero}.}
Structural modification has been investigated in the case of (not strong) structural controllability of systems defined by zero/nonzero pattern matrices~\cite{zhang2019minimal,chen2018minimal,wang2023input}. This is a relaxed version of strong structural controllability, and naturally, this formulation is not directly extendable to our strong structural controllability setting. Therefore, motivated by the recent advances in strong structural controllability~\cite{jia2020unifying}, we look at the problem of making minimal changes to the input matrix of a structured system to guarantee its strong structural controllability. 
Our contributions are as follows:
\begin{itemize}[leftmargin=0.3cm]
    \item  In \Cref{sec:probformulation}, we formulate the minimal (input) structural modification problem to ensure strong structural controllability. We discuss the conditions for the feasibility of the problem in \Cref{thm:feasibility}, and under the feasibility conditions, we present a more intuitive reformulation via \Cref{thm:eq_opt}. 
    \item In \Cref{sec:greedy}, we present and analyze a simple greedy algorithm for the structural modification problem. \Cref{prop:greedyopt} shows that although a greedy solution typically performs well in general, in certain cases, the cost function of the greedy solution can be arbitrarily larger than the optimal solution.  
    \item In \Cref{sec:mcmc}, we devise a Markov chain Monte Carlo (MCMC)-based solution, with guarantees on optimality (\Cref{thm:steadystate}) and error probability (\Cref{prop:prob_error}) if it is run sufficiently long.
\end{itemize}
Overall, our results provide interesting insights and devise design algorithms to achieve strong structural controllability. 

\section{Notation and Preliminaries}\label{sec:prelim}
A pattern matrix $\cl{M}\in\{0,*,?\}^{p\times q}$ defines a pattern class, $\bb{P}(\cl{M})\subseteq\bb{R}^{p\times q}$ as follows.
\begin{multline}\label{eq:patternclass}
    \bb{P}(\cl{M}) = \lc\bs{M}\in\bb{R}^{p\times q}: 
    \bs{M}_{ij}=0\; \text{if }\; \cl{M}_{ij}=0 \rd\\\ld \; \text{and} \;
    \bs{M}_{ij}\neq  0\;\text{if }\; \cl{M}_{ij}=*\rc.
\end{multline}
We note that the entries in $\bs{M}\in\bb{P}(\cl{M})$ corresponding to $?$ entries in $\cl{M}$ can either be zero or nonzero (arbitrary).

    A structured system $(\cl{A},\cl{B})$ is a class of linear dynamical systems $(\bs{A},\bs{B})$ whose state matrix $\bs{A}\in\bb{P}(\cl{A})$ and input matrix $\bs{B}\in\bb{P}(\cl{B})$, where the set $\bb{P}(\cdot)$ is defined in \eqref{eq:patternclass}. The system $(\cl{A},\cl{B})$ is called strongly structurally controllable if the linear dynamical system $(\bs{A}, \bs{B})$ is controllable for any $\bs{A}\in\bb{P}(\cl{A})$ and $\bs{B}\in\bb{P}(\cl{B})$.
 It can be tested using the notion of the rank of pattern matrices, 
$\rank(\cl{M}) = \min_{\bs{M}\in\bb{P}(\cl{M})}\rank(\bs{M})$.
     In particular, the pattern matrix $\cl{M}$ is full row rank if every matrix $\bs{M}\in\bb{P}(\cl{M})$ has full row rank. A rank-based test for strong structural controllability of a structured system is as follows.
 \begin{theorem}[\cite{jia2020unifying}]\label{thm:strongcontrol}
     The system $(\cl{A},\cl{B})$ is strong structural controllable if and only if the following two conditions hold:
 \begin{enumerate}
     \item The pattern matrix $\begin{bmatrix}
         \cl{A} & \cl{B}
     \end{bmatrix} $ has full row rank.
\item The pattern matrix $\begin{bmatrix}
         \cl{Q}(\cl{A}) & \cl{B}
     \end{bmatrix}$ has full row rank, where  pattern matrix $\cl{Q}(\cl{A})$ is 
\begin{equation}\label{eq:mod_A}
    [\cl{Q}(\cl{A})]_{ij} = \begin{cases}
    \cl{A}_{ij} &\text{if }\;  i\neq j\\
        * &\text{if }\; i=j \;\text{and}\; \cl{A}_{ii}=0\\
        ? &\text{otherwise}.
    \end{cases}
\end{equation}
 \end{enumerate}
 \end{theorem} 
Here, a graph-theoretic algorithm, summarized in \Cref{alg:colorchange}~\cite{jia2020unifying} can be used to test the pattern matrix rank.  
\begin{algorithm}[hptb]
\caption{Color change rule}
\label{alg:colorchange}
\begin{algorithmic}[1]
\Statex \textit {\bf Input:} Pattern matrix $\cl{M}\in\{0,*,?\}^{p\times q}$
\State Set $\bb{E}_*=\{(j,i): \cl{M}_{ij}=*\}$; $\bb{E}_?=\{(j,i): \cl{M}_{ij}=?\}$
\State \label{step:initalize}Initialize the set of white vertices $W\gets\{1,2,\ldots,p\}$
\Repeat
\State Set $W_{\rm del} \gets \lc i\in W: \exists j \;\text{such that}\; (j, i) \in E_* \rd$
\Statex \hfill $\ld \text{and}\; (j,\tilde{i})\notin\bb{E}_*\cup \bb{E}_?, \forall \tilde{i}\in W\setminus\{i\}\rc$
\State \label{step:update_W}Color vertex set $W_{\rm del}$ black, $W\gets W\setminus W_{\rm del}$ 
\Until {$W_{\rm del}=\emptyset$}
\State Set $W(\cl{M})\gets W$
\Statex \textit{\bf Output:} Set of white vertices $W(\cl{M})$
\end{algorithmic}
\end{algorithm}

\begin{theorem}[\cite{jia2020unifying}]\label{thm:colorchange}
A given pattern matrix $\cl{M}\in\{0,{*,?}\}^{p\times q}$ is full row rank if and only if the set of white vertices $W(\cl{M})$ outputted by the color change rule in \Cref{alg:colorchange} is empty.     
\end{theorem}

Further, we define a zero forcing set of a pattern matrix.

\begin{defn}\label{defn:ZF}
    For a pattern matrix $\cl{M}\in\{0,*,?\}^{ p\times q}$, the set $V^{0}\subseteq\{1,2,\ldots,p\}$ is called its zero forcing set if \Cref{alg:colorchange} returns an empty set $W(\cl{M})$ when initilization in Step~\ref{step:initalize} is changed to $W\gets\{1,2,\ldots,p\}\setminus V^0$.
\end{defn}
Relying on the above results on strong structural controllability, the next section presents our problem formulation.

\section{Minimal Structural Modification Problem}\label{sec:probformulation}
Consider a strongly structurally uncontrollable system $(\bar{\cl{A}}\in\{0,*,?\}^{n\times n},\bar{\cl{B}}\in\{0,*,?\}^{n\times m})$ as given in \Cref{sec:prelim}. We address the problem of making minimum changes to the pattern matrix $\bar{\cl{B}}$ such that the resulting system is strongly structurally controllable. To quantify the change in the pattern matrix, we use the Hamming distance metric, $\operatorname{dist}(\cl{B},\bar{\cl{B}}) = \lv\{ (i,j): \cl{B}_{ij}\neq\bar{\cl{B}}_{ij}\}\rv$. Further, using \Cref{thm:strongcontrol,thm:colorchange}, the resulting optimization problem is  
    \begin{multline}\label{eq:problem_colorchange}
\underset{\cl{B}\in\{0,*,?\}^{n \times m}}{\arg\min}\;\operatorname{dist}(\cl{B},\bar{\cl{B}})\\\text{s. t.}\; 
        W([\bar{\cl{A}} \;\; \cl{B}])\cup W([\cl{Q}(\bar{\cl{A}})\;\; \cl{B}])=\emptyset,
    \end{multline}
where $W(\cdot)$ is \Cref{alg:colorchange}'s output and $\cl{Q}(\cdot)$ is as in \eqref{eq:mod_A}.

 {We next discuss the feasibility conditions of the problem, using the notion of zero forcing set in \Cref{defn:ZF}.
\begin{theorem}\label{thm:feasibility}
    Consider a given structured system $(\bar{\cl{A}}\in\{0,*,?\}^{n\times n},\bar{\cl{B}}\in\{0,*,?\}^{n\times m})$ and $\cl{Q}(\bar{\cl{A}})$ as defined in \eqref{eq:mod_A}. The structural modification problem \eqref{eq:problem_colorchange} is feasible only if  
    \begin{equation}\label{eq:nece}
        m\geq \max\{Z(\bar{\cl{A}}),Z(\cl{Q}(\bar{\cl{A}}))\},
    \end{equation}
    where $Z(\cdot)$ is the zero forcing number i.e., the minimum cardinality $|V^0|$ over all the zero forcing sets $V^0$. Also, the problem \eqref{eq:problem_colorchange} is feasible if 
    \begin{multline}\label{eq:suff}
        m\geq Z^{\mathrm{joint}} = \min_{V\subset\{1,2,\ldots,n\}} |V|\\
        \text{s. t.}\; V\;\text{is a common zero forcing set of}\; \bar{\cl{A}} \;\text{and} \; \cl{Q}(\bar{\cl{A}}).
    \end{multline}
\end{theorem}
\begin{proof}
    See \Cref{app:feasibility}.
\end{proof} 
The necessary condition is not always sufficient and vice versa. For example, consider the following pattern matrices,
\begin{equation}
    \bar{\cl{A}}^{(1)}=\begin{bmatrix}
    0 & ? & ?\\ 0 & * & ? \\ 0 & ? & ?
\end{bmatrix}\!, \bar{\cl{A}}^{(2)}=\begin{bmatrix}
    0 & 0 & 0\\ 0 & * & 0 \\ 0 & 0 & *
\end{bmatrix} \text{and} \; \cl{B}^{*}=\begin{bmatrix}
    * & 0 \\ * & 0 \\  0 & *
\end{bmatrix}\!.
\end{equation}
For $\bar{\cl{A}} = \bar{\cl{A}}^{(1)}$, the term $\max\{Z(\bar{\cl{A}}),Z(\cl{Q}(\bar{\cl{A}}))\}=2$, yet the problem \eqref{eq:problem_colorchange} is infeasible if $m= 2$. For $\bar{\cl{A}}=\bar{\cl{A}}^{(2)}$, we have $Z^{\mathrm{joint}}=3$, but with $m=2$, the matrix $\cl{B}^{*}$ is feasible.
} So, the feasibility set of \eqref{eq:problem_colorchange} is hard to obtain and the following theorem reformulates the original problem in \eqref{eq:problem_colorchange} to another equivalent feasible problem.
\begin{theorem}\label{thm:eq_opt}
    For a structured system $(\bar{\cl{A}}\in\{0,*,?\}^{n\times n},\bar{\cl{B}}\in\{0,*,?\}^{n\times m})$, the the structural modification problem  in \eqref{eq:problem_colorchange}, {if feasible}, is equivalent to following optimization problem,
    \begin{equation}
    \label{eq:cost_opt}
\underset{\cl{B}\in\bb{B}}{\arg\min}\; c(\cl{B}).
    \end{equation}
Here, the feasible set $\bb{B}$ and the cost function $c(\cdot)$ are
\begin{align}\label{eq:feasibleset}
\bb{B} &= \lc \cl{B}\in\{0,*,?\}^{n \times m}: \cl{B}_{ij}\neq ?\; \forall(i,j) \;\text{with}\; \bar{\cl{B}}_{ij}\neq ?\rc\\
c(\cl{B}) &= \operatorname{dist}(\cl{B},\bar{\cl{B}}) + \epsilon \lb \lv W([\bar{\cl{A}} \;\; \cl{B}])\rv+\lv W([\cl{Q}(\bar{\cl{A}})\;\; \cl{B}])\rv\rb,\label{eq:costfnt}
\end{align}
    where the set $W(\cdot)$ is the output of \Cref{alg:colorchange}, $\cl{Q}(\cdot)$ is defined in \eqref{eq:mod_A}, and the constant $\epsilon>nm$.
\end{theorem}
\begin{proof}
    See \Cref{app:eq_opt}.
\end{proof}
{The cost function ensures that all feasible solutions satisfy $c(\cl{B})<\epsilon$ (see details in \Cref{app:feasibility}), i.e., the solution to \eqref{eq:cost_opt} ensures structural controllability only if the corresponding cost $c(\cl{B})<\epsilon$. Also, from \Cref{thm:feasibility}, we derive that the optimal solution changes at least $\max\{Z(\bar{\cl{A}}),Z(\cl{Q}(\bar{\cl{A}}))\}$ and at most $Z^{\mathrm{joint}}$ columns. Thus, the optimal cost $c^*$ satisfies
\begin{equation}
   \max\{Z(\bar{\cl{A}}),Z(\cl{Q}(\bar{\cl{A}}))\}\leq  c^*\leq nZ^{\mathrm{joint}}
\end{equation}} Moreover, if $\bar{\cl{B}}$ does not have any arbitrary entries (i.e.,$?$ entries), then $\bb{B} = \{0,*\}^{n \times m}$ and the optimal solution also does not have any $?$ entries. Hence, our formulation and algorithms directly apply to the setting without $?$ entries.


\section{Structural Modification Algorithms}\label{sec:algos}
We present two algorithms to solve the minimal structural modification, starting with a greedy approach.

\subsection{Greedy Algorithm}\label{sec:greedy}
To design the greedy algorithm, we note that every iteration of \Cref{alg:colorchange} considers the sub-pattern matrix $\cl{M}^{W}$ of the input pattern matrix $\cl{M}$ restricted to the rows indexed by the set of white nodes $W$. The algorithm removes an element $i$ from the set $W$ only if the sub-pattern matrix  $\cl{M}^{W}$ has a column $j$ with one $*$ entry and zeros elsewhere. Then, the $i$th row of $\cl{B}$ corresponding to the $*$ entry in the $j$th column of $\cl{M}^{W}$ is removed from $W$. Therefore, in the next iteration, the $j$th column of $\cl{M}^{W}$ has all zeros and can not induce more color changes. Consequently, once a column of the input pattern matrix induces a color change, it can not induce any other color change in the subsequent iterations. Based on these observations, our greedy algorithm iteratively changes one column of the current iterate $\cl{B}$ that is locally optimal in each iteration. Also, once a column of $\cl{B}$ is changed, it is kept fixed in the subsequent iterations. So, in every iteration of the greedy algorithm, we first compute the set of the white nodes in $[\bar{\cl{A}} \;\; \cl{B}]$ and $[\cl{Q}(\bar{\cl{A}})\;\; \cl{B}]$ using \Cref{alg:colorchange} and the previous iterate $\cl{B}$, i.e.,
\begin{equation}\label{eq:control_cost}
    I = W([\bar{\cl{A}} \;\; \cl{B}])\cup W([\cl{Q}(\bar{\cl{A}})\;\; \cl{B}]).
\end{equation}
Then, the greedy algorithm solves for the optimal column of $\cl{B}$ restricted to rows indexed by $I$, 
\begin{equation}\label{eq:greedy_opt}
    (i^*\!,j^*)\!=\!\underset{\substack{i\in I, j\in J}}{\arg\min}\; c(\tilde{\cl{B}})\;
    \text{s. t.}\; \tilde{\cl{B}}_{\tilde{i}\tilde{j}}=\!
    \begin{cases}
        * &\!\!\text{if}\;\!(\tilde{i},\tilde{j})=(i,j)\\
        0 &\!\! \text{if}\;\!\tilde{j}=j, \; \tilde{i}\in I\!\setminus\! \{i\} \\
        \cl{B}_{\tilde{i}\tilde{j}} &\!\! \text{otherwise},
    \end{cases}
\end{equation}
where $J$ is the index of columns of $\cl{B}$ {that are identical} to the corresponding columns of $\bar{\cl{B}}$, i.e., unchanged in the previous iterations.  Here, $j^*$ is the column changed in the current iteration, and $i^*$ denotes the location of $*$ in the $j$th column of $\cl{B}$. The greedy algorithm stops when {the feasibility set of \eqref{eq:greedy_opt} is empty, i.e., $I=\emptyset$ or $J=\emptyset$}. The overall greedy algorithm is summarized in \Cref{alg:greedy}.

\begin{algorithm}[hptb]
\caption{Greedy structural modification}
\label{alg:greedy}
\begin{algorithmic}[1]
\Statex \textit {\bf Input:} System $(\bar{\cl{A}}\in\{0,*,?\}^{n\times n},\bar{\cl{B}}\in\{0,*,?\}^{n\times m})$
\State Initialize $\cl{B}=\bar{\cl{B}}$ and $J=\{1,2,\ldots,m\}$
\State Compute $I$ using \eqref{eq:control_cost}
\While {{$I\neq \emptyset $ and $J\neq\emptyset$ }}
\State $\cl{B}_{i^*j^*}\gets *$ and $\cl{B}_{\tilde{i}j^*}\gets 0$ for $\tilde{i}\neq I\setminus\{i^*\}$ using \eqref{eq:greedy_opt}
\State Update $I$ using \eqref{eq:control_cost} and $J\gets J\setminus\{j^*\}$
\EndWhile
\Statex \textit{\bf Output:} Modified input pattern matrix $\cl{B}\in\{0,*,?\}^{n\times m}$
\end{algorithmic}
\end{algorithm}

The greedy algorithm is simple to implement, but it does not guarantee the solution's optimality. The following proposition presents {a case} where the cost returned by the greedy solution can be arbitrarily larger than the optimal cost. 
\begin{prop}\label{prop:greedyopt}
    For any given $\delta>0$, there exists integers $n,m>0$ and a structured system $(\bar{\cl{A}}\in\{0,*,?\}^{n\times n},\bar{\cl{B}}\in\{0,*,?\}^{n\times m})$ such that the solution $\cl{B}_{{\rm greedy}}$ returned by the greedy algorithm in \Cref{alg:greedy}  satisfies
    \begin{equation}\label{eq:greedy_bnd}
        c(\cl{B}_{{\rm greedy}}) > \delta \ls\min_{\cl{B}\in\{0,*,?\}^{n \times m}}\; c(\cl{B})\rs.
    \end{equation}
\end{prop}
\begin{proof}
    See \Cref{app:greedyopt}.
\end{proof}

Since the worst-case performance of the greedy algorithm is not bounded, we propose another algorithm for structural modification that relies on MCMC.
\subsection{Monte Carlo Markov Chain Algorithm}\label{sec:mcmc}
MCMC is a powerful stochastic optimization technique used to solve discrete optimization problems. The underlying principle of this approach is to randomly generate pattern matrices from $\bb{B}$ using a probability distribution and return the sample with the lowest cost. A common technique to define the probability distribution is to use the softmax function to favor the pattern matrices with smaller $c(\cl{B})$,
\begin{equation}\label{eq:desired_dist}
    p_T(\cl{B}) =  e^{-c(\cl{B})/T}/G,
\end{equation}
where $G=\sum_{\cl{B}'\in\bb{B}}e^{-c(\cl{B}')/T}$ is the normalization constant of the distribution. As $T$ gets smaller, for $\cl{B}\notin\underset{\cl{B}\in\bb{B}}{\arg\min}\; c(\cl{B})$,
\begin{equation}\label{eq:T_limit}
    \lim_{T\to 0}  p_T(\cl{B}) \!= \!\lim_{T\to 0}  \frac{ e^{-c(\cl{B})/T}}{\sum_{\cl{B}'\in\bb{B}}e^{-c(\cl{B}')/T}} = 0.
\end{equation}
Therefore, we arrive at the optimal solution when $T$ is small. Nonetheless, computing the distribution is cumbersome because the number of candidate solutions increases exponentially with $nm$. To solve this problem, we build a discrete-time Markov chain (DTMC), which converges to its stationary distribution equal to the desired distribution $p_T$ in \eqref{eq:desired_dist}. If we simulate the DTMC for a sufficiently long period for a small value of $T$, its state arrives at an optimal solution. 

In the following, we construct a DTMC $\{\cl{B}(t)\in\bb{B}\}_{t>0}$ whose stationary distribution is $p_T(\cl{B})$.
Since the state space is of size $|\bb{B}|$ from \Cref{thm:eq_opt}, the DTMC is defined by the one-step probability transistion matrix $\bs{P}_T\in[0,1]^{|\bb{B}|\times |\bb{B}|}$. Let $\cl{B}$ be the current state. We allow only transitions to the neighboring states that differ from the current state $\cl{B}$ by at most one entry, i.e., the set of next states is $\bb{S}(\cl{B})\cup\{\cl{B}\}$ where
\begin{equation}\label{eq:searchspace}
    \bb{S}(\cl{B}) = \{\cl{B}'\in\bb{B}: \operatorname{dist}(\cl{B},\cl{B}')=1\}.
\end{equation}
Therefore, we arrive at
\begin{equation}\label{eq:tpm_zero}
    \bs{P}_T(\cl{B},\cl{B}') =0, \;\forall\; \cl{B}'\notin \bb{S}(\cl{B})\cup\{\cl{B}\}.
\end{equation}

{ We define the matrix $\bs{P}_T$ such that any neighboring states in $\bb{S}(\cl{B})$ is equally probable.} First,  we choose $(i,j)$ from $\{1,2,\ldots,n\}\times \{1,2,\ldots,m\}$ obeying the distribution $d(i,j)$, 
\begin{equation}\label{eq:sample_dist}
    d(i,j) = \begin{cases}
        1/|\bb{S}(\cl{B})|  & \text{ if } \bar{\cl{B}}_{ij}\neq ?\\
        2/|\bb{S}(\cl{B})|  & \text{ if } \bar{\cl{B}}_{ij}= ?,
    \end{cases}
\end{equation}
where $\bb{S}(\cl{B})$ is defined in \eqref{eq:searchspace} and 
its size is  $|\bb{S}(\cl{B})|=nm+|\{(i,j):\bar{\cl{B}}_{ij}= ?\}|<2nm$. {Then, we replace the $(i, j)$-th entry $\cl{B}_{ij}$ of $\cl{B}$ with a sample uniformly randomly chosen from $\{0,*\}\cup\{\bar{\cl{B}}_{ij}\}\setminus \{\cl{B}_{ij}\}$, to get the next state $\cl{B}'$.} We note that if $\bar{\cl{B}}_{ij}= ?$, the entry $\cl{B}_{ij}'$ has two choices, and when $\bar{\cl{B}}_{ij}\neq  ?$, the entry $\cl{B}_{ij}'$ has only one choice. So the above process chooses $\cl{B}'$ uniformly at random from $\bb{S}(\cl{B}(t))\subset\bb{B}$.

Further, we also assign a non-zero probability to continue in the current state. If $c(\cl{B})>c(\cl{B}')$, the DTMC jumps to the neighboring state $\cl{B}'$ with a lower cost with probability one. Also, if $c(\cl{B})<c(\cl{B}')$, the DTMC jumps to the neighboring state $\cl{B}'$ with a higher cost with probability $e^{[c(\cl{B})-c(\cl{B}')]/T}$. 
The resulting DTMC transition probabilities are as follows. 
\begin{equation}\label{eq:tpm_nonzero}
    \bs{P}_T(\cl{B},\cl{B}') = \frac{1}{|\bb{S}(\cl{B})|}\min\{1,e^{[c(\cl{B})-c(\cl{B}')]/T}\}\;\forall\; \cl{B}'\in\bb{S}(\cl{B}).
\end{equation}
Since $\sum_{\cl{B}'\in\bb{B}}\bs{P}_T(\cl{B},\cl{B}')=1$, we deduce 
\begin{equation}\label{eq:tpm_diag}
    \bs{P}_T(\cl{B},\cl{B}) \!=\!1-\frac{1}{|\bb{S}(\cl{B})|}\displaystyle\sum_{\cl{B}'\in\bb{S}(\cl{B})}\min\{1,e^{[c(\cl{B})-c(\cl{B}')]/T}\},
\end{equation}
from \eqref{eq:tpm_zero} and \eqref{eq:tpm_nonzero}.
The resulting algorithm is summarized in \Cref{alg:MCMC}. Here, we decrease $T$ in every $R_{\max}$ iteration. The initial large value for $T$ allows the MCMC to jump between the states quickly for a flexible random search. Later, we use small values to converge to the unique steady state distribution, as guaranteed by the following theorem.
\begin{algorithm}[hptb]
\caption{MCMC-based structural modification}
\label{alg:MCMC}
\begin{algorithmic}[1]
\Statex \textit {\bf Input:} System $(\bar{\cl{A}}\in\{0,*,?\}^{n\times n},\bar{\cl{B}}\in\{0,*,?\}^{n\times m})$
\State Initialize the parameters $R_{\max},T_{\rm stop}, \alpha<1$
\State Set $\epsilon=nm+1$, {$T > T_{\rm stop}$}, and $\cl{B} = \bar{\cl{B}}$
\While{$T_{\rm stop} \leq T$}
\For {$r=1,2,\ldots,R_{\max}$}
    \State Set $\cl{B}' \gets \cl{B}$
    \State Generate $(i,j)$ from the distribution $d$ in \eqref{eq:sample_dist}
    \State  \parbox[t]{7cm}{Choose $\cl{B}_{ij}'$ from $\{0,*\}\cup\{\bar{\cl{B}}_{ij}\}\setminus\{\cl{B}_{ij}\}$ uniformly at random }
    \State $\cl{B}  \gets \cl{B}'$ with probability $\min\{1,e^{[c(\cl{B})-c(\cl{B}')]/T}\}$
    \EndFor
    \State $T \gets T\alpha$ 
\EndWhile
\Statex \textit{\bf Output:} Input pattern matrix $\cl{B}$
\end{algorithmic}
\end{algorithm}
\begin{theorem}\label{thm:steadystate}
For any fixed $T>0$, the DTMC with states from $\bb{B}$ in \eqref{eq:feasibleset} and transition probabilities given by \eqref{eq:tpm_zero}, \eqref{eq:tpm_nonzero}, and \eqref{eq:tpm_diag} admits a unique steady state distribution given by  \eqref{eq:desired_dist}.
\end{theorem}
\begin{proof}
    {The proof is similar to the proof of \cite[Lemma 2 and Theorem 3]{yashashwi2019minimizing}, and hence, omitted.} 
\end{proof}
    
The above theorem establishes that when $T$ approaches $0$ and $R_{\max}$ goes to $\infty$, the pattern matrix returned by \Cref{alg:MCMC} is a solution to \eqref{eq:problem_colorchange} almost surely due to \eqref{eq:T_limit}. {However, in practice, $T\neq 0$, leading to an error in estimation, as characterized by the following result.
\begin{prop}\label{prop:prob_error}
    For any $T<\infty$, the MCMC algorithm in \Cref{alg:MCMC} converges to an optimal solution of \eqref{eq:cost_opt} with probability exceeding $\frac{B^*}{\lv\bb{B}\rv}$, where we define
    \begin{equation}
        B^* = \lv\underset{\cl{B}\in\bb{B}}{\arg\min}\; c(\cl{B})\rv<\lv\bb{B}\rv.
    \end{equation}
    Further, for any $\delta<1-\frac{B^*}{\lv\bb{B}\rv}$, the MCMC algorithm arrives at an optimal solution with probability $1-\delta$ if 
 \begin{equation}\label{eq:Tstopbound}
     T<\frac{1}{\log  (1/\delta-1)+\log(\lv\bb{B}\rv/B^*-1)}.
 \end{equation}
\end{prop}
\begin{proof}
    See \Cref{app:prob_error}.
\end{proof}
\Cref{prop:prob_error} depends on $B^*$, the number of optimal solutions of \eqref{eq:control_cost}, which is hard to compute. So, we choose the lower bound of $B^*\geq 1$ to compute $T_{\mathrm{stop}}$ for a desired error probability $\delta$, as follows:
\begin{align}
    T_{\mathrm{stop}}&<\frac{1}{\log  (1/\delta-1)+\log(\lv\bb{B}\rv-1)}\\
    &\leq\frac{1}{\log  (1/\delta-1)+\log(\lv\bb{B}\rv/B^*-1)}.
\end{align}
}
\section{Numerical Results}\label{sec:numericalresults}
For numerical evaluation, we choose the MCMC parameters $R_{\max} = 50000$, $T_{\mathrm{Stop}} = 10^{-10}$,  $\alpha=10^{-1}$, and $T=1$ as the starting value. In \Cref{fig:Changes}, the cost is always less than $nm<\epsilon$, indicating that the resulting system is strongly structurally controllable. So, the cost is the number of changes required to make the system controllable{\footnote{{Our code is available at \url{https://github.com/Gethuj/2024LCSS_StructralModification}}  }}.

\begin{figure*}[!htbp]
	\centering
	\begin{subfigure}[t]{0.3\textwidth}
		\centering
		\includegraphics[height=0.7\textwidth]{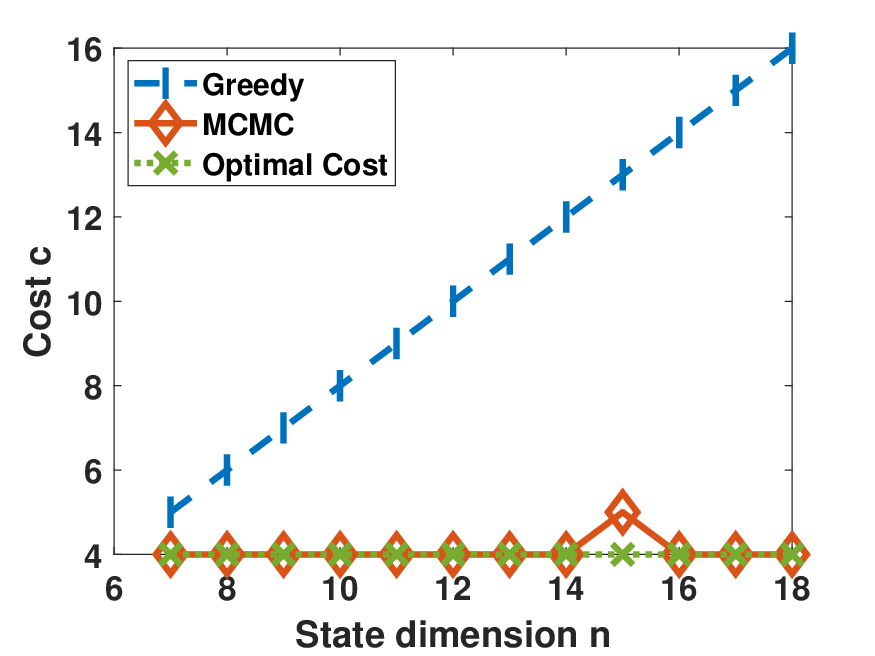}
		\caption{Example discussed in the proof of \Cref{prop:greedyopt} where $n=m$}
		\label{fig:CornerCase}
	\end{subfigure}
	\begin{subfigure}[t]{0.6\textwidth}
		\centering
		\includegraphics[height=0.35\textwidth]{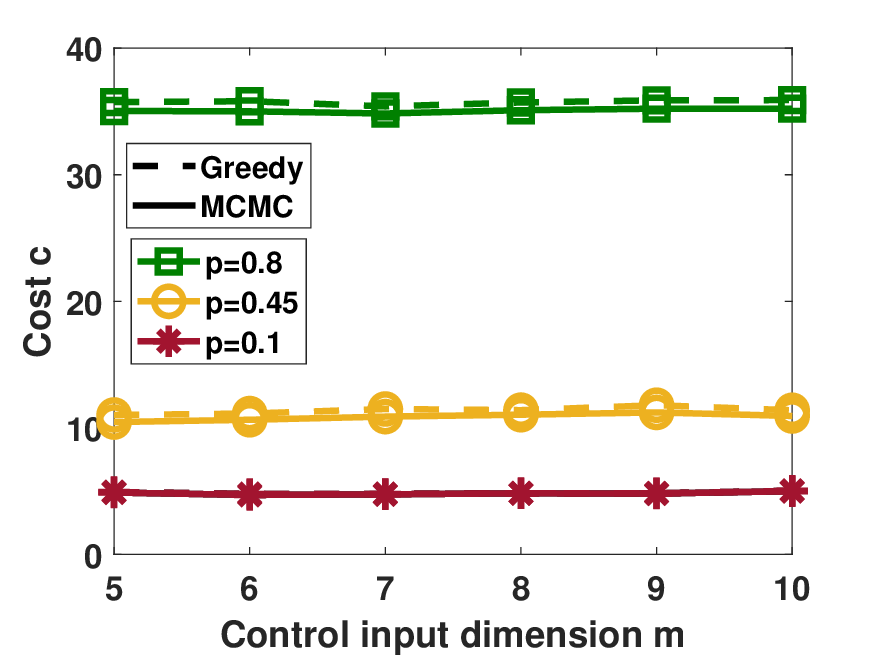}
            \includegraphics[height=0.35\textwidth]{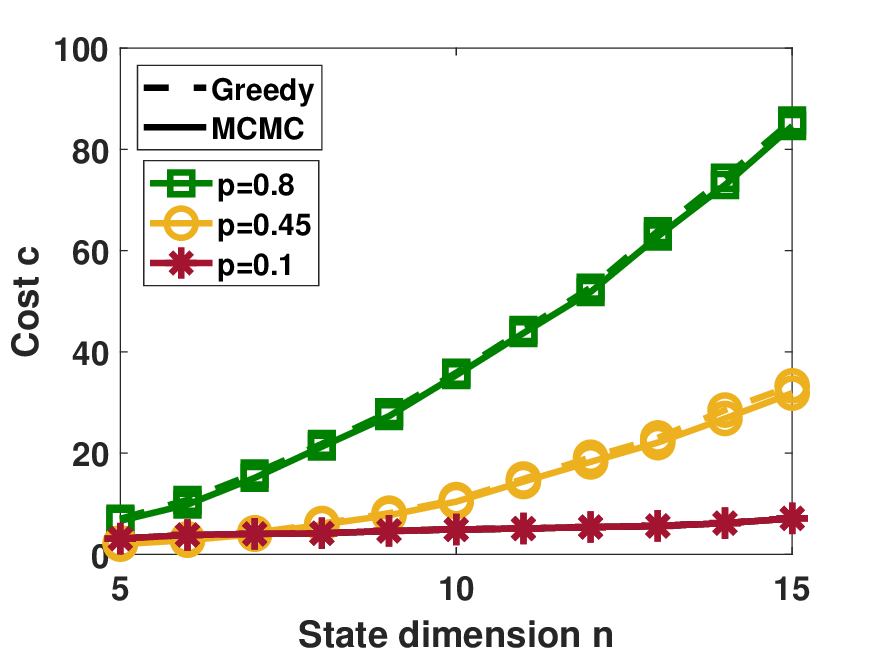}
		\caption{\ER~graph with $p$ as the probability of a $*$ entry}
		\label{fig:ER}
	\end{subfigure}\hspace{0.5cm}
	\caption{Variation of the number of changes required by the greedy and MCMC algorithms with the state and control input dimensions. For the special case discussed in the proof of \Cref{prop:greedyopt}, the greedy algorithm performed poorly while the MCMC algorithm returned the optimal cost in most cases. However,  for \ER~graphs, the greedy algorithm has a performance comparable to that of the MCMC algorithm. }
	\label{fig:Changes}
\end{figure*}

For the example discussed in the proof of \Cref{prop:greedyopt}, \Cref{fig:CornerCase} shows that the greedy solution's cost is $n-2$, whereas the MCMC algorithm returns the optimal solution in most cases. Next, we look at the average performance of the algorithms (over 100 trials) when the adjacency matrix of \ER~graphs is used to generate both $\bar{\cl{A}}$ and $\bar{\cl{B}}$. \Cref{fig:ER} shows the algorithms' performances with $p=0.1, 0.45,$ and $0.8$, where $p$ is the probability of having an edge between any two nodes in the graph ($*$ entry). Also, the existence of an edge in the graph is unknown ($?$ entry) with probability $0.1$. \Cref{fig:ER} indicates that the greedy and MCMC algorithms return similar solutions (mostly, the difference is less than $1$), illustrating that the greedy algorithm performs well
generally. 

From \Cref{fig:Changes}, the number of changes increases with the state dimension $n$, as expected. Further, if the system $(\bar{\cl{A}},\cl{B}')$ is controllable for a submatrix $\cl{B}'$ of a matrix $\cl{B}$, the system $(\bar{\cl{A}},\cl{B})$ is also controllable. So, the algorithm needs to change only a submatrix of $\bar{\cl{B}}$, making the required number of changes insensitive to the control dimension $m$ (\Cref{fig:ER}). We also infer that the number of changes increases with $p$, the probability of an $*$ entry. It is intuitive as a triangular structure of $\bs{B}$ favors strong structural controllability. So we need to make fewer changes if $\bar{\cl{B}}$ has a lot of zeros.

{Our experiments also indicated that the greedy algorithm runs four-order faster than the MCMC algorithm (details omitted due to lack of space). Nonetheless, we highlight that simplistic greedy approaches can yield arbitrarily poor solutions as shown in \Cref{prop:greedyopt}.}
\section{Conclusion}
We addressed the problem of making minimal changes to the input pattern matrix of a structured system to ensure strong structural controllability. We offered a greedy algorithm and an MCMC-based solution with provable guarantees. Our results open new interesting directions for future work, such as strong structural controllability under restricted structural modifications, its robustness to perturbations in the state matrices, and extensions to time-varying systems.
\appendices
\crefalias{section}{appendix}
{\section{Proof of \Cref{thm:feasibility}}\label{app:feasibility}
To prove the necessary conditions, consider \Cref{alg:colorchange} with $\cl{M}=[\bar{\cl{A}}\;\;\cl{B}]$ as the input, for some $\cl{B}\in\{0,*,?\}^{n\times m}$. In every iteration, \Cref{alg:colorchange} removes an entry $i$ from the set $W$ if there exists a column $j$ such that $\cl{M}_{ij}=*$ and $\cl{M}_{\tilde{i}j}=0$ for all $\tilde{i}\in W\!\setminus\!\{i\}$. Therefore, in the next iteration, the $j$th column of $\cl{M}$ has all zeros corresponding to rows indexed by $W$ and can not induce any more color changes. Thus, a column of the input matrix induces at most one color change. Further, by the definition of zero forcing number, \Cref{alg:colorchange} does not return an empty set unless $Z(\bar{\cl{A}})$ columns of $\cl{B}$ induces a color change. Hence, if there is a feasible solution, the number of columns $m$ of  $\cl{B}$ satisfy $m\geq Z(\bar{\cl{A}})$. Using similar arguments with $\cl{M}=[\cl{Q}(\bar{\cl{A}})\;\;\cl{B}]$, we deduce that $m\geq Z(\cl{Q}(\bar{\cl{A}})$. Hence, the minimal structural modification problem \eqref{eq:problem_colorchange} is feasible only if \eqref{eq:nece} holds.

Next, to establish the sufficient condition in \eqref{eq:suff}, it is enough to construct a matrix $\cl{B}^*\in\{0,*,?\}^{n\times Z^{\mathrm{joint}}}$ such that $W([\bar{\cl{A}} \;\; \cl{B}^*])\cup W([\cl{Q}(\bar{\cl{A}})\;\; \cl{B}^*])=\emptyset$. Let $V$ be a minimal zero forcing set of both $\bar{\cl{A}}$ and $\cl{Q}(\bar{\cl{A}})$ and $\cl{B}^*$ be a submatrix of the diagonal matrix with $*$ along the diagonal, formed by $Z^{\mathrm{joint}}$ columns indexed by $V$. Then, the first iteration of \Cref{alg:colorchange} with $[\bar{\cl{A}}\;\;\cl{B}^*]$ as input removes $V$ from $W$. So we deduce that $W([\bar{\cl{A}}\;\;\cl{B}^*])$ is the same as $W(\bar{\cl{A}})$ when initialization in Step~\ref{step:initalize} of \Cref{alg:colorchange} is changed to $W\gets\{1,2,\ldots,p\}\setminus V$. Further, from \Cref{defn:ZF}, we conclude that $W([\bar{\cl{A}}\;\;\cl{B}^*])=\emptyset$. Similarly, we can show that $W([\cl{Q}(\bar{\cl{A}})\;\;\cl{B}^*])=\emptyset$. Therefore, if $m\geq Z^{\mathrm{joint}}$, \eqref{eq:problem_colorchange} is feasible for any system $(\bar{\cl{A}},\bar{\cl{B}})$.}

\section{Proof of \Cref{thm:eq_opt}}\label{app:eq_opt}
We prove the result in two steps {under the assumption that \eqref{eq:problem_colorchange} is feasible}. We first show that, using \eqref{eq:costfnt}, the following optimization is equivalent to \eqref{eq:problem_colorchange},
\begin{equation}\label{eq:costfnt_opt}
    \underset{\cl{B}\in\{0,*,?\}^{n\times m}}{\arg\min}\; c(\cl{B}).
\end{equation}
In the second step, {we show that, if} $\cl{B}\notin \bb{B}$, then $\cl{B}$ {is not an optimal solution} of \eqref{eq:costfnt_opt}. This step establishes that \eqref{eq:costfnt_opt} is equivalent to $\underset{\cl{B}\in\bb{B}}{\arg\min}\; c(\cl{B})$, and the proof is complete.

We start with the first step. Let $w(\cl{B}) = |W([\bar{\cl{A}} \;\; \cl{B}])|+| W([\cl{Q}(\bar{\cl{A}})\;\; \cl{B}])|$, for any $\cl{B}\in\{0,*,?\}^{n \times m}$. Then, $c(\cl{B}) = \operatorname{dist}(\cl{B},\bar{\cl{B}}) + \epsilon w(\cl{B})$. Since $w(\cl{B})\geq 1$ when $w(\cl{B})\neq 0$ and $\epsilon>nm$, we have 
\begin{align}\label{eq:eq_optpf_1}
    \underset{\cl{B}: w(\cl{B})\neq 0}{\min}\; c(\cl{B}) 
    &> \underset{\cl{B}: w(\cl{B})\neq 0}{\min}\; \operatorname{dist}(\cl{B},\bar{\cl{B}}) + nm\\
   &\geq nm \geq \underset{\cl{B}: w(\cl{B})= 0}{\min}\; \operatorname{dist}(\cl{B},\bar{\cl{B}}) \label{eq:ineq}\\
   &\geq \underset{\cl{B}: w(\cl{B})= 0}{\min}\; c(\cl{B}),
\end{align}
where \eqref{eq:ineq} follows because $0\leq \operatorname{dist}(\cl{B},\bar{\cl{B}})\leq nm,$ for any $\cl{B},\bar{\cl{B}}\in\{0,*,?\}$.
Therefore, we obtain
\begin{equation}\label{eq:eq_optpf_2}
    \underset{\cl{B}}{\arg\min}\; c(\cl{B})  = \underset{\cl{B}: w(\cl{B})= 0}{\arg\min}\; c(\cl{B})= \underset{\cl{B}: w(\cl{B})= 0}{\arg\min}\;\operatorname{dist}(\cl{B},\bar{\cl{B}}).
\end{equation}
Here, $w(\cl{B})= 0$ if and only if $W([\bar{\cl{A}} \;\; \cl{B}])\cup W([\cl{Q}(\bar{\cl{A}})\;\; \cl{B}])$ is an empty set. Hence, problems \eqref{eq:problem_colorchange} and \eqref{eq:costfnt_opt} are equivalent.

Now, we present the second step using a proof by negation. Suppose there exists  $\tilde{\cl{B}}\notin\bb{B}$ such that $\tilde{\cl{B}}$ is a solution to \eqref{eq:costfnt_opt}, i.e., there exists $(i,j)$ for which $\tilde{\cl{B}}_{ij}=?$ and $\bar{\cl{B}}_{ij}\neq?$. Due to the equivalence of \eqref{eq:problem_colorchange} and \eqref{eq:costfnt_opt}, $\tilde{\cl{B}}$ also belongs to the solution set of \eqref{eq:problem_colorchange}. Therefore, $(\bar{\cl{A}},\tilde{\cl{B}})$ is strongly structurally controllable. Further, let pattern matrices $\tilde{\cl{B}}^{(0)},\tilde{\cl{B}}^{(*)}\in\{0,*,?\}^{n\times m}$ be such that they are identical to $\tilde{\cl{B}}$ except that $\tilde{\cl{B}}^{(0)}_{ij}=0$ and $\tilde{\cl{B}}^{(*)}_{ij}=*$. 
  However, from \eqref{eq:patternclass} we note that 
  \begin{equation}
      \bb{P}(\tilde{\cl{B}}) =  \bb{P}(\tilde{\cl{B}}^{(0)}) \cup \bb{P}(\tilde{\cl{B}}^{(*)}).
  \end{equation}
  Hence, we see that $(\bar{\cl{A}},\tilde{\cl{B}}^{(0)})$ and $(\bar{\cl{A}},\tilde{\cl{B}}^{(*)})$ are strongly structurally controllable. However, since either $\tilde{\cl{B}}^{(0)}_{ij}$ or $\tilde{\cl{B}}^{(*)}_{ij}$ is the same as $\bar{\cl{B}}_{ij}$ and $\tilde{\cl{B}}_{ij}\neq \bar{\cl{B}}_{ij}$, we arrive at
  \begin{equation}
      c(\tilde{\cl{B}}) = \operatorname{dist}(\tilde{\cl{B}},\bar{\cl{B}}) > \min \lc c(\tilde{\cl{B}}^{(0)}) ,c(\tilde{\cl{B}}^{(*)}) \rc.
  \end{equation}
  Thus, the assumption that $\tilde{\cl{B}}$ minimizes the cost $c$ does not hold, and the proof is complete. 
  \section{Proof of \Cref{prop:greedyopt}}\label{app:greedyopt}
Consider a system $(\bar{\cl{A}},\bar{\cl{B}})$ where $\bar{\cl{B}}=\bs{0}\in\{0,*,?\}^{n\times n}$ with $n\geq 6$ 
 and $\bar{\cl{A}}\in\{0,*,?\}^{n\times n}$ is defined as follows:
\begin{equation}
    \bar{\cl{A}}_{ij} = \begin{cases}
        ? & \text{if}\; i=j\\
        0 & \text{if}\; 1\leq i,j\leq n-3\; \text{and}\; |i-j|>1\\
        0 & \text{if}\; n-2\leq i,j\leq n\;\text{and}\; |i-j|>1\\
        * & \text{otherwise}.
    \end{cases}
\end{equation}
Here, $\cl{Q}(\bar{\cl{A}})=\bar{\cl{A}}$, and 
 $   c(\cl{B})= \operatorname{dist}(\cl{B},\bar{\cl{B}}) + 2\epsilon \lv W([\bar{\cl{A}} \;\; \cl{B}])\rv$.
Let $\cl{B}^*$ be a submatrix of the diagonal matrix with $*$ along the diagonal, formed by columns indexed by $\{1,n-2,n-1,n\}$. 
Then, $ \lv W([\bar{\cl{A}} \;\; \cl{B}^*])\rv=0$, making $\cl{B}^*$ feasible, and we get
\begin{equation}\label{eq:lowerbnd_cost}
    \min_{\cl{B}\in\{0,*,?\}^{n \times m}}\; c(\cl{B})\leq  c(\cl{B}^*)=4.
\end{equation} Further, if we apply a greedy algorithm, in the first iteration, changing any entry of $\cl{B}$ to $*$ returns the cost $1+2\epsilon(n-1)$. Let the algorithm changes $\cl{B}_{11}$ to $*$. In the next iteration, changing any entry of $\cl{B}$, except from the first row and first column to $*$ reduces the cost to $c(\cl{B}) = 2+2\epsilon(n-2)$. Assume that the algorithm changes $\cl{B}_{22}$ to $*$. Similarly, after $n-2$ iterations, the algorithm terminates with $\cl{B}_{{\rm greedy}}$ such that $[\cl{B}_{{\rm greedy}}]_{ii}=*$ for $i=1,2,\ldots,n-2$ and zeros elsewhere. Thus, the greedy solution's cost is
\begin{equation}
    c(\cl{B}_{{\rm greedy}}) = n-2 \geq  \frac{n-2}{4}\ls\min_{\cl{B}\in\{0,*,?\}^{n \times m}}c(\cl{B})\rs,
\end{equation}
where we use \eqref{eq:lowerbnd_cost}. So, for any $\delta>0$, if we choose $n>4\delta+2$, the lower bound \eqref{eq:greedy_bnd} holds.
{\section{Proof of \Cref{prop:prob_error}}\label{app:prob_error}
Let the optimal cost of the minimal structural modification problem in \eqref{eq:cost_opt} be $c^*=\underset{\cl{B}\in\bb{B}}{\min}\; c(\cl{B})$.  Then, for any $T>0$, from the DTMC distribution in \eqref{eq:desired_dist} and the union bound, the probability of the DTMC converging to an optimal solution is $B^*e^{-c^*/T}/G$. Since $c(\cl{B}')\geq c^*+1$ for all $\cl{B}'\notin \underset{\cl{B}\in\bb{B}}{\arg\min}\; c(\cl{B})$, we have
 \begin{align}
     B^*e^{-c^*/T}/G&\geq\frac{B^*e^{-c^*/T}}{B^*e^{-c^*/T}+(\lv\bb{B}\rv-B^*)e^{-(c^*+1)/T}}\\
     &=\frac{B^*}{B^*+(\lv\bb{B}\rv-B^*)e^{-1/T}}>\frac{B^*}{\lv\bb{B}\rv}.
 \end{align} 
 Therefore, we prove the first part of the result. Furthermore, the error  probability exceeds $1-\delta$ if $\frac{B^*}{B^*+(\lv\bb{B}\rv-B^*)e^{-1/T}}>1-\delta$. Rearranging this relation, we arrive at \eqref{eq:Tstopbound}.}

\bibliographystyle{IEEEtran}
\bibliography{StructuralModifications}
\end{document}